\documentclass[a4paper,english,numberwithinsect,nolineno]{socg-lipics-v2021}
\pdfoutput=1
\hideLIPIcs

\EventEditors{Wolfgang Mulzer and Jeff M. Phillips}
\EventNoEds{2}
\EventLongTitle{40th International Symposium on Computational Geometry
(SoCG 2024)}
\EventShortTitle{SoCG 2024}
\EventAcronym{SoCG}
\EventYear{2024}
\EventDate{June 11-14, 2024}
\EventLocation{Athens, Greece}
\EventLogo{socg-logo}
\SeriesVolume{293}
\ArticleNo{XX}     % <-- This will be filled in by the typesetters

\usepackage{microtype,xspace,mathtools} % if unwanted, comment out

\title{A Clique-Based Separator for Intersection Graphs of Geodesic Disks in $\Reals^2$}
\author{Boris Aronov}{Department of Computer Science and Engineering, Tandon School of Engineering, New York University, Brooklyn, NY 11201 USA \and \url{https://engineering.nyu.edu/faculty/boris-aronov}}{boris.aronov@nyu.edu}{https://orcid.org/0000-0003-3110-4702}{Work has been supported by NSF grant CCF~20-08551.}

\author{Mark de Berg}{Department of Mathematics and Computer Science, TU Eindhoven, the Netherlands}{M.T.d.Berg@tue.nl}{https://orcid.org/0000-0001-5770-3784}{MdB is supported by the  Dutch Research Council (NWO) through Gravitation-grant NETWORKS-024.002.003.}

\author{Leonidas Theocharous}{Department of Mathematics and Computer Science, TU Eindhoven, the Netherlands}{l.theocharous@tue.nl}{https://orcid.org/0000-0002-1707-6787}{LT is supported by the  Dutch Research Council (NWO) through Gravitation-grant NETWORKS-024.002.003.}

\authorrunning{B.~Aronov, M.~de Berg, and L.~Theocharous}
\titlerunning{A Clique-Based Separator for Geodesic Disks in Polygonal Domains}

\Copyright{Boris Aronov, Mark de Berg, and Leonidas Theocharous}
\relatedversion{A preliminary version of this work will appear in \emph{SoCG'24}\cite{socg-version}}

\newcommand{\BeginMyItemize}{\begin{itemize}\setlength{\itemsep}{-\parskip}}
\newcommand{\EndMyItemize}{\end{itemize}}

\newcommand{\BeginMyEnumerate}{\begin{enumerate}\setlength{\itemsep}{-\parskip}}
\newcommand{\EndMyEnumerate}{\end{enumerate}}

\newcommand{\mypara}[1]{\medskip\noindent{\sf\textbf{#1}}}  %more "standard" paragraph enviroment
\renewcommand{\leq}{\leqslant}
\renewcommand{\geq}{\geqslant}

\newtheorem{defin}{Definition}
\newtheorem{lem}[defin]{Lemma}
\newtheorem{myfact}[defin]{Fact}

\newenvironment{myquote}%
  {\list{}{\leftmargin=4mm\rightmargin=4mm}\item[]}%
  {\endlist}
\newcommand{\claiminproof}[2]{\begin{myquote}\noindent\emph{Claim~#1.} #2 \end{myquote}}
\newcommand{\obsinproof}[2]{\begin{myquote}\noindent\emph{Observation~#1.} #2 \end{myquote}}
\newenvironment{proofinproof}{\begin{myquote}\noindent\emph{Proof.}}{\hfill $\lhd$ \end{myquote}}

\newcommand{\Reals}{\mathbb{R}}

\newcommand{\bd}{\partial}

\newcommand{\diam}{\mathrm{diam}}

\newcommand{\graph}{\ensuremath{\mathcal{G}}}

\newcommand{\D}{\ensuremath{\mathcal{D}}}
\newcommand{\E}{\ensuremath{\mathcal{E}}}

\newcommand{\X}{\ensuremath{\mathcal{X}}}

\newcommand{\G}{\ensuremath{\mathcal{G}}}

\newcommand{\calS}{\ensuremath{\mathcal{S}}}

\newcommand{\eps}{\varepsilon}
\newcommand{\mydef}{\coloneq}
\newcommand{\etal}{\emph{et al.}\xspace}

\newcommand{\ig}{\graph^{\times}}

\newcommand{\sep}{\mathcal{S}}

\DeclareMathOperator{\ply}{ply}

\newcommand{\hdist}{\mathit{hdist}}
\newcommand{\mis}{\textsc{Maximum Independent Set}\xspace}

\newcommand{\fvs}{\textsc{Feedback Vertex Set}\xspace}
\newcommand{\col}{\textsc{Coloring}\xspace}

\begin{document}

\maketitle

\begin{abstract}
Let $d$ be a (well-behaved) shortest-path metric defined on a path-connected subset of $\Reals^2$ and let $\D=\{D_1,\ldots,D_n\}$ be a set of
geodesic disks with respect to the metric~$d$. We prove that $\ig(\D)$, the intersection
graph of the disks in $\D$, has a clique-based separator consisting of $O(n^{3/4+\eps})$ cliques.
This significantly extends the class of objects whose intersection graphs have small clique-based separators.

Our clique-based separator yields an algorithm for $q$-\col, which runs in time $2^{O(n^{3/4+\eps})}$, assuming the boundaries of the disks $D_i$ can be computed in polynomial time. We also use our clique-based separator to obtain a simple, efficient, and almost exact
distance oracle for intersection graphs of geodesic disks. Our distance oracle
uses $O(n^{7/4+\eps})$ storage and can report the hop distance between any two nodes in $\ig(\D)$
in $O(n^{3/4+\eps})$ time, up to an additive error of one.
So far, distance oracles with an additive error of one that use subquadratic storage and sublinear
query time were not known for such general graph classes.
\end{abstract}

%---------------------------------------------------------------------------------
\section{Introduction}
%---------------------------------------------------------------------------------

%---------------------------------------------------------------------------------
\subparagraph{(Clique-based) separators.} 
%---------------------------------------------------------------------------------
The Planar Separator Theorem states that any planar graph with $n$ nodes has a 
\emph{balanced separator} of size~$O(\sqrt{n})$. In other words, for any planar graph
$\graph=(V,E)$ there exists a subset $S\subset V$ of size $O(\sqrt{n})$ with the 
following property:\footnote{Hereafter we will simply refer to such a subset as
a \emph{separator}, omitting the word ``balanced'' for brevity. Moreover, we do not require
the balance factor to be 2/3, but we allow
the size of the subsets $A$ and $B$ to be at most $\delta n$, for some fixed constant~$\delta<1$. }
$V\setminus S$ can be split into subsets $A$ and $B$ with 
$|A|\leq 2n/3$ and $|B|\leq 2n/3$ such that there are no edges between $A$ and $B$.
This fundamental result was first proved in 1979 by Lipton and Tarjan~\cite{LT-planar-separator-thm} 
and has been simplified and refined in several ways, see,~e.g.,~\cite{DV-cycle-separator,fox2008}.
% \mdb{Add a few references to other papers on planar graphs, for example on weighted version or that give a cycle separator.} 
%\leo{i see, the first paper is about cycle separators. As a second, we could put this \cite{sariel-separator} by Sariel }
% \mdb{Instead of Sariel's paper I would perhaps add the paper by Djidjev, H., Venkatesan that we cite in our previous paper.} 
%\leo{isn't this \cite{DV-cycle-separator}? we could add miller's paper maybe}
The Planar Separator Theorem 
proved to be extremely useful for obtaining efficient divide-and-conquer algorithms 
for a large variety of problems on planar graphs.

In this paper we are interested in \emph{geometric intersection graphs} in the plane. 
These are graphs whose node set corresponds to a set $\D$ of objects in the plane and that
have an edge between two nodes iff the corresponding objects intersect. We will denote
this intersection graph by $\ig(\D)$. \emph{(Unit) disk graphs}
and \emph{string graphs}---where the set $\D$ consists of (unit) disks and curves, respectively---are
among the most popular types of intersection graphs. Unit disk graphs in particular 
have been studied extensively, because they serve as a model for wireless communication networks.
It is well known that disk graphs are generalizations of planar graphs, because by
the Circle Packing Theorem (also known as the Koebe–Andreev–Thurston Theorem)
every planar graph is the intersection graph of a set of disks with disjoint interiors~\cite{thurston-78}.
It is therefore natural to try to extend the Planar Separator Theorem to 
intersection graphs. A direct generalization is clearly impossible,
however, since intersection graphs can contain arbitrarily large cliques.

There are several ways to still obtain separator theorems for intersection graphs.
One can allow the size of the separator to depend on $m$, the number 
of edges, instead of on the number of vertices. For example, any
string graph admits a separator of size~$O(\sqrt{m})$~\cite{Lee-string-sep}. One may 
also put restrictions on the set $\D$ to prevent large cliques. For example, 
if we restrict the \emph{density}\footnote{The density of a set $\D$ 
of objects in $\Reals^d$ is the maximum, over all balls~$b\subset \Reals^d$, of the cardinality of
the sets $\{ o\in \D : o\cap b\neq \emptyset \mbox{ and } \diam(o)\geq \diam(b) \}$, where
$\diam(o)$ denotes the diameter of object~$o$.} of a set $\D$ of planar objects 
to some value $\lambda$, then we can obtain a separator\footnote{The result generalizes to $\Reals^d$, where the bound becomes~$O(\lambda^{1/d} n^{1-1/d})$.}
of size $O(\sqrt{\lambda n})$~\cite{HQ-low-dens-sep}.
% \leo{correct bound?} \mdb{Yes}
This result implies that intersection graphs of disks (or, more generally, of \emph{fat} objects)
whose ply is bounded by a constant admit a separator of size~$O(\sqrt{n})$.
(The \emph{ply} of a set of objects in~$\Reals^2$ is the maximum, over all
points $q\in\Reals^2$, of the number of objects containing~$q$.)
Recently, De~Berg~\etal~\cite{bbkmz-ethf-20}, following work of Fu~\cite{FU2011379}, introduced what we
will call \emph{clique-based separators}. A clique-based separator 
is a (balanced) separator consisting of cliques, and its size is 
not measured as the total number of nodes of the cliques, but as the number of cliques.
De~Berg~\etal showed that the intersection graph of a set $\D$ of convex fat object in the plane
admits a separator consisting of $O(\sqrt{n})$ cliques. In fact, their result is
slightly stronger: it states that there is a separator $S$ that can be decomposed
into 
% \ba{explain $\eps$.  is it arbitrarily small positive, with implied constant depending on it?} 
% \mdb{Decided not to talk about implied constant, since this is standard.}
$O(\sqrt{n})$ cliques $C_1,\ldots,C_t$ such that  $\sum_{i=1}^t \gamma(|C|) = O(\sqrt{n})$,
where $\gamma(t)=\log t$ is a cost function\footnote{In fact, the result holds for any cost function
$\gamma(t)=O(t^{\frac12 -\eps})$, where $\eps>0$ is a constant, but in typical applications either only the number of cliques is important, or $\gamma(t)=\log t$ suffices.} on the cliques.
Clique-based separators are useful because cliques can be handled very efficiently
for many problems. De~Berg~\etal used their clique-based separators
to develop a unified framework to solve many classic graph problems, including
\textsc{Independent Set}, \textsc{Dominating Set}, \textsc{Feedback Vertex Set}, and more.
% \mdb{Mention applications by others? Paper has been cited more than 50 times.}

Recently De~Berg~\etal~\cite{bkmt-cbsgis-23} proved clique-based separator theorems for various
other classes of objects, including map graphs and intersection graphs of pseudo-disks,
which admit clique-based separators consisting of $O(\sqrt{n})$ and $O(n^{2/3})$ cliques, respectively.
They also showed that intersection graphs of \emph{geodesic disks inside a simple polygon}---that
is, geodesic disks induced by the standard shortest-path metric inside the polygon---admit
a clique-based separator consisting of $O(n^{2/3})$ cliques.
They left the case of geodesic disks in a polygon with holes as an open problem.
We note that string graphs, which subsume the class of intersection graphs of geodesic disks,
do \emph{not} admit clique-based separators of sublinear size, since string
graphs can contain arbitrarily large induced bipartite cliques.
%\boris{do we mean ``INDUCED bipartite cliques''?}
 
%---------------------------------------------------------------------------------
\subparagraph{Our results.}
%---------------------------------------------------------------------------------
In this paper we show that intersection graphs of geodesic disks in a polygon with holes
admit a clique-based separator consisting of a sublinear number of cliques. 
Our result is actually much more general,
as it shows that \emph{for any well-behaved shortest-path metric~$d$} defined
on a path-connected and closed subset~$F\subset \Reals^2$, a set of geodesic disks 
with respect to that metric admits a separator consisting of $O(n^{3/4+\eps})$ cliques. 
This includes the shortest-path metric defined by a set of (possibly curved) obstacles in the plane,
the shortest-path metric defined on a terrain, and the shortest-path metric
among weighted regions in the plane. (See Section~\ref{se:proper-path-system} for the
formal requirements on a well-behaved shortest-path metric.)
%\mdb{What properties do we need from the shortest-path metric? Should we try to formulate this? }
Note that we do not require shortest paths to be unique, nor do we put a bound on
the number of intersections between the boundaries of two geodesic disks, nor do we 
require the metric space to have bounded doubling dimension.

The generality of our setting implies that previous approaches to construct
clique-based separators will not work. For example, the method of
De~Berg \etal~\cite{bbkmz-ethf-20}
to construct a clique-based separator for Euclidean disks 
crucially relies on disks being fat objects. More precisely, it
uses the fact that many relatively large fat objects intersecting
a given square must form $O(1)$ cliques. This property
already fails for geodesic disks for the Euclidean shortest-path metric 
inside a simple polygon. Thus, De~Berg~\etal~\cite{bkmt-cbsgis-23} use a different approach:
they prove that geodesic disks inside a simple polygon
are \emph{pseudo-disks} and then show how to obtain a clique-based separator for pseudo-disks. 
In our more general setting, however, geodesic disks need not be pseudo-disks:
the boundaries of two geodesic disks in an arbitrary metric can intersect 
arbitrarily many times---this is already the case for the Euclidean shortest-path
metric in a polygon with holes. Thus we have to proceed differently.

The idea of our new approach is as follows. We first reduce the ply of the set $\D$ by removing 
all  cliques of size $\Omega(n^{1/5})$, thus obtaining a set $\D^*$ of ply $O(n^{1/5})$. 
(The removed cliques will eventually be added to the separator.
A similar preprocessing step to reduce the ply was used 
in~\cite{bkmt-cbsgis-23} to handle pseudo-disks.) The remaining arrangement can 
still be arbitrarily complex, however. To overcome this, we ignore the arrangement 
induced by the disks, and instead focus on the realization of the
graph $\ig(\D^*)$ obtained by drawing a shortest path $\pi_{ij}$ between
the centers of any two intersecting disks $D_i,D_j\in\D^*$.
We then prove that the number of edges of $\ig(\D^*)$ must be $O(n^{8/5})$;
otherwise there will be an intersection point of two shortest paths
$\pi_{ij},\pi_{k\ell}$ that has large ply, which is not possible due to the
preprocessing step. Since the number of edges of $\ig(\D^*)$ is $O(n^{8/5})$
we can use the separator result on string graphs to obtain a
separator of size $O(\sqrt{n^{8/5}}) = O(n^{4/5})$. Adding the cliques that were removed
in the preprocessing step then yields a separator consisting of $O(n^{4/5})$ cliques
and singletons. Finally, we devise a bootstrapping
mechanism to further decrease the size of the separator, leading to a
separator consisting of $O(n^{3/4+\eps})$ cliques.

Clique-based separators give sub-exponential algorithms for 
\mis,
\fvs, and $q$-\col for constant~$q$~\cite{bkmt-cbsgis-23}.  When using
our clique-based separator, the running times for \mis and \fvs are
inferior to what is known for string graphs. 
For $q$-\col with $q\geq 4$ this is not the case, however,
since $4$-\col does not admit a sub-exponential algorithm, assuming {\sc eth}~\cite{BonnetR19}.
Our clique-based separator, on the other hand, yields an algorithm with running time $2^{O(n^{3/4+\eps})}$,
assuming the boundaries of the disks $D_i$ can be computed in polynomial time. 
Another application of our separator result is to distance oracles, as discussed next.

%---------------------------------------------------------------------------------
\subparagraph{Application to distance oracles.}
%---------------------------------------------------------------------------------
One of the most basic queries one can ask about a (possibly edge-weighted) graph~$G=(V,E)$ 
is a \emph{distance query}: given two nodes $s,t\in V$, what is the distance between them? 
A data structure answering such queries is called a \emph{distance oracle}. 
One can simply store all pairwise distances in a distance matrix so that distance queries can 
be answered in $O(1)$ time, but this requires $\Theta(n^2)$ storage,
where $n \mydef |V|$. At the other extreme, one can 
do no preprocessing and answer a query by running a single-source shortest path
algorithm; this uses $O(n+m)$ storage, where $m \mydef |E|$ and a query 
costs~$\Omega(n)$ time. The main question is if sublinear query time
can be achieved with subquadratic storage. 

As an application of our clique-based separator we present a simple 
and almost exact distance oracle for intersection graphs of geodesic
disks. While our approach is standard---it was already used in 1996 by 
Arikati~\etal~\cite{DBLP:conf/esa/ArikatiCCDSZ96}---it does provide an almost exact distance
oracle in a more general setting than what was known before, thus showing the
power of clique-based separators.  Our distance oracle 
uses $O(n^{7/4+\eps})$ storage and can report the hop distance between any two nodes in $\ig(\D)$
in $O(n^{3/4+\eps})$ time, up to an additive error of one. 
This is the first distance oracle with only an additive error for
a graph class that is more general than planar graphs:
even for Euclidean unit-disk graphs the known distance oracle has
a multiplicative error~\cite{chan-skrepetos,gao-zhang}. (Admittedly, the storage and 
query time of this oracle are significantly better than of  ours.)
We note that clique-based separators were recently also used (in a similar way) to
obtain a distance oracle for so-called transmission graphs~\cite{b-nrdo-23}.
(A transmission graph is a directed graph whose nodes correspond to Euclidean
disks in the plane, and that have an edge from disk $D_i$ to disk $D_j$ iff
the center of $D_j$ is contained in $D_i$.) 
%---------------------------------------------------------------------------------
\subparagraph{Related work on distance oracles.}
%---------------------------------------------------------------------------------
There has been a lot of work on distance oracles, both for undirected and for directed graphs;
see for example the survey by Sommer~\cite{DBLP:journals/csur/Sommer14} for an overview
of the work up to~2014. Since our interest is in intersection graphs,
which are undirected, hereafter we will restrict our discussion to undirected graphs.

In the edge-weighted setting, one obviously needs to store all edge weights
to answer queries exactly. Hence, for exact distance oracles on weighted graphs the focus has been on 
on sparse graphs and, in particular, on planar graphs. 
% The currently best exact distance oracle with $O(\log(1/\eps))$ query time still uses $O(n^{5/3+\eps})$  storage~\cite{DBLP:conf/isaac/Fredslund-Hansen21}. 
Very recently, Charalampopoulos~\etal~\cite{DBLP:journals/jacm/CharalampopoulosGLMPWW23}
achieved a major breakthrough by providing an exact distance oracle for planar graphs
that uses $O(n^{1+o(1)})$ storage and has $O(\log^{2+o(1)})$ query time. The solution 
even works for weighted planar graphs, and it allows other trade-offs as well.
For $(1+\eps)$-approximate distance oracles on planar graphs, 
Le and Wulff-Nilsen~\cite{DBLP:conf/focs/LeW21} showed how to achieve $O(1/\eps^2)$ query time with $O(n/\eps^2)$ storage.
We refer the reader to the recent paper by Charalampopoulos~\etal~\cite{DBLP:journals/jacm/CharalampopoulosGLMPWW23}
for a historical overview of the results on planar distance oracles. 

For non-planar graphs the results are far less good: the distance oracles are
approximate and typically use significantly super-linear storage;
see the survey by Sommer~\cite{DBLP:journals/csur/Sommer14}. 
% \mdb{Anything appearing after 2014?} \leo{didn't find something}
For example, Chechik~\cite{DBLP:conf/stoc/Chechik14} presented a $(2k-1)$-approximate distance oracle
using $O(kn^{1+1/k})$ storage and with $O(1)$ query time, for any given integer $k\geq 1$. 
Moreover, for $t<2k-1$, any $t$-approximate distance oracle requires $\Omega(kn^{1+1/k})$ 
bits of storage~\cite{DBLP:journals/jacm/ThorupZ05}.

Somewhat better results are known for 
unweighted graphs, where the approximation often has an additive term in addition to the
multiplicative term. 
More precisely, if $d_G(s,t)$ denotes the actual
distance between $s$ and $t$ in $G$, then an \emph{$(\alpha,\beta)$-approximate oracle} 
reports a distance $d^*$ such that $d_G(s,t) \leq d^* \leq \alpha\cdot d_G(s,t)+\beta$.
% \leo{shouldn't it be the latter? can the reported distance be smaller?}
% \mdb{It indeed seems to be defined that way, so I changed it.}
Patrascu and Roditty~\cite{patrascu} presented a $(2,1)$-approximate oracle with $O(1)$ query time
that uses $O(n^{5/3})$ storage, 
Abraham and Gavoille~\cite{abraham-gavoille} presented a $(2k-2,1)$-approximate oracle with $O(k)$ query time
that uses $\tilde{O}(n^{1+2/(2k-1)})$ storage (for $k\geq 2$).

To summarize, for non-planar graphs all distance oracles with subquadratic storage
have a multiplicative error of at least~2, even in the unweighted case. 
The only exception is for the rather restricted case of unit-disk graphs, where 
Gao and Zhang~\cite{gao-zhang}, and later Chan and Skrepetos~\cite{chan-skrepetos},
presented a $(1+\eps)$-approximate distance oracle with $O(1)$ query time 
that uses $O_{\eps}(n\log n)$ storage. (This result actually also
works in the weighted setting, where the weight of an edge between two disks is the
Euclidean distance between their centers.)

%---------------------------------------------------------------------------------------
\section{A clique-based separator for geodesic disks in \texorpdfstring{$\Reals^2$}{ℝ²}}\label{clique-based}
%---------------------------------------------------------------------------------------
Let $d$ be a metric defined on a closed path-connected subset $F\subset \Reals^2$, and let $\D=\{D_1,\ldots,D_n\}$
be a set of geodesic disks in~$F$, with respect to the metric~$d$. Thus each disk\footnote{From now on, 
to simplify the terminology, we omit the adjective \emph{geodesic} and simply use the term \emph{disk}
to refer to the objects in~$\D$.}~$D_i$ is
defined as $D_i \mydef \{ q\in F: d(q,p_i) \leq r_i\}$, where $p_i\in F$ is the center of $D_i$
an~$r_i\geq 0$ is its radius. Let $\D_0 \mydef \D$ and
recall that $\ig(\D_0)$ denotes the intersection graph
of $\D_0$. (The reason for introducing the notation $\D_0$ will become clear shortly.)
We denote the set of edges of $\ig(\D_0)$ by~$E$ and define $m\mydef |E|$.

%---------------------------------------------------------------------------------------
\subsection{A first bound}
%---------------------------------------------------------------------------------------
Our basic construction of a clique-based separator~$\sep$ for $\ig(\D_0)$, yields a separator with $O(n^{4/5})$ cliques. In Section~\ref{sec:improving},
we will apply a bootstrapping scheme to further reduce the size of the separator.
The basic construction proceeds in three steps: in a preprocessing step we reduce the
ply of the set of disks we need to deal with, in the second step we prove that if
the ply is sublinear then the number of edges in the intersection graph is subquadratic,
and in the third step we construct the separator.

%---------------------------------------------------------------------------------------
\mypara{Step~1: Reducing the ply.} 
%---------------------------------------------------------------------------------------
For a point $p\in F$, let $\D_0(p) \mydef \{D_i \in \D_0: p\in D_i\}$ be the set of disks from~$\D_0$
containing~$p$---note that $\D_0(p)$ forms a clique in $\ig(\D_0)$---and define $\ply(p) \mydef |\D_0(p)|$ 
to be the ply of~$p$ with respect to~$\D_0$. The ply of the set $\D_0$ is defined as 
$\ply(\D_0) \mydef \max \{p\in F: \ply(p)\}$.

We start by reducing the ply of $\D_0$ in the following greedy manner. Let $\alpha$ be
a fixed constant with $0<\alpha<1$. In the basic construction we will use $\alpha=1/5$,
but in our bootstrapping scheme we will work with other values as well.  
We check whether there exists a point $p$ such that $|\D_0(p)|\geq \frac{1}{4}n^{\alpha}$.
If so, we remove $\D_0(p)$ from $\D_0$ and add it as a clique to~$\sep$.  
We repeat this process until $\ply(\D_0)<\frac{1}{4}n^{\alpha}$. 
Thus in the first step at most $4n^{1-\alpha}$ cliques are added to~$\sep$.

To avoid confusion with our initial set $\D_0$, we denote the set of disks remaining at the end
of Step~1 by~$\D_1$ and we denote the set of edges of $\ig(\D_1)$ by $E_1$. 

%---------------------------------------------------------------------------------------
\mypara{Step~2: Bounding the size of~$E_1$.} 
%---------------------------------------------------------------------------------------
To bound the size of $E_1$, we draw a shortest path $\pi_{ij}$ between the centers $p_i,p_j$
of every two intersecting 
disks $D_i,D_j\in \D_1$, thus obtaining a geometric realization of the graph~$\ig(\D_1)$.  
From now on, with a slight abuse of notation, we will not distinguish between an 
edge $(D_i,D_j)$ in $\ig(\D_1)$ and its geometric realization~$\pi_{ij}$. 
Let $\Pi(\D_1) \mydef \{ \pi_{ij} : (D_i,D_j) \in E_1\}$ be the resulting set of paths.
To focus on the main idea behind our proof we will, for the time being, 
assume that $\Pi(\D_1)$ is a \emph{proper path set}: a set of paths such that
any two paths $\pi_{ij},\pi_{k\ell}$ have at most two points in common,
each intersection point is either a shared endpoint or a proper crossing,
and no proper crossing coincides with another proper crossing or with an endpoint.
When $\pi_{ij}$ and $\pi_{k\ell}$ have a proper crossing, we say that they \emph{cross},
otherwise they do not cross.

We need the following well-known result about the number of crossings
in dense graphs, known as the Crossing Lemma~\cite{AJTAI19829, l-civlsi-03}. The term
\emph{planar drawing} in the Crossing Lemma refers to a drawing where no edge interior
passes through a vertex and all
intersections are proper crossings. Thus it applies to a proper path set. 
%---------------------------------------------------------------------------------------
\begin{lemma}[Crossing Lemma]
There exists a constant $c>0$, such that every planar drawing of a graph with $n$ vertices
and $m\geq 4n$ edges contains at least $c\frac{m^3}{n^2}$ crossings.
\end{lemma}
%---------------------------------------------------------------------------------------
Using the Crossing Lemma we will show that $|E_1|=O(n^{\frac{3+\alpha}{2}})$, as follows.
If $|E_1|=O(n^{\frac{3+\alpha}{2}})$ does not hold, then by the Crossing Lemma 
there must be many crossings between the  edges~$\pi_{ij}\in E_1$. We will
show that this implies that there is a crossing of ply greater than~$\frac{1}{4}n^{\alpha}$, 
thus contradicting that $\ply(\D_1)<\frac{1}{4}n^{\alpha}$. We now make this idea precise.
%---------------------------------------------------------------------------------------
\begin{lemma} \label{subquadratic}
Let $\ig(\D_1)=(\D_1,E_1)$ be the intersection graph of a set $\D_1$ of disks such
that $\ply(\D_1)<\frac{1}{4}n^{\alpha}$. Then $|E_1|\leq \sqrt{\frac{4}{c}} \cdot n^{\frac{3+\alpha}{2}}$, 
where $c$ is the constant appearing in the Crossing Lemma.
\end{lemma}
%---------------------------------------------------------------------------------------
%\begin{proof} 
%
\begin{figure}
\begin{center}
\includegraphics{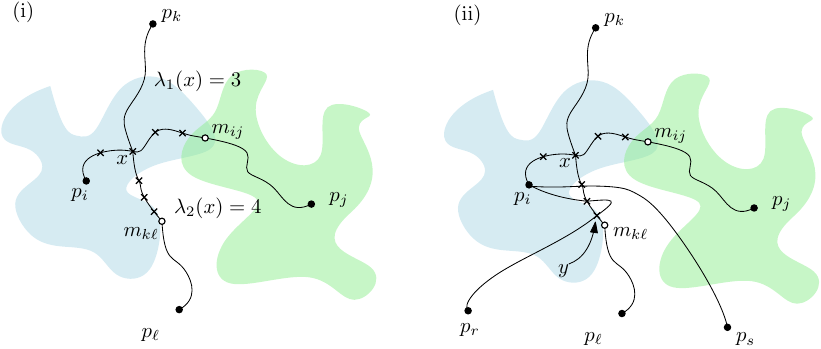}
\end{center}
\caption{(i) An example of a labeling for a crossing $x\in \X$. (ii) Here, the crossing $y\in \pi_{ij}[x,m_{k\ell}]$ is assigned to $D_i$ a total of four times. } 
\label{fig:labeling}
\end{figure}

\begin{proof}
Consider the proper path set $\Pi(\D_1)$ and let $\X$ be the set of crossings between 
the paths in~$\Pi(\D_1)$.
Assume for a contradiction that $|E_1|>\sqrt{\frac{4}{c}}\cdot n^{\frac{3+\alpha}{2}}$.
We will show that then there has to exist a crossing $x\in \X$ of ply at 
least $\frac{n^{\alpha}}{4}$, which contradicts that $\ply(\D_1)<\frac{1}{4}n^{\alpha}$.

We start by giving a lower bound on the total ply of all crossings in the drawing. 
To this end, we split each edge $\pi_{ij}\in E_1$ in two \emph{half-edges} as follows.
For two points $x,y\in\pi_{ij}$, let $\pi_{ij}[x,y]$ denote the subpath of $\pi_{ij}$
between~$x$ and~$y$. Recall that $p_i$ is the center of disk~$D_i$. 
We now pick an arbitrary point $m_{ij} \in \pi_{ij} \cap \left( D_i\cap D_j \right)$ and split
$\pi_{ij}$ at $m_{ij}$ into a half-edge $\pi_{ij}[p_i,m_{ij}]$ connecting $p_i$ to~$m_{ij}$
and a half-edge~$\pi_{ij}[p_j,m_{ij}]$ connecting $p_j$ to~$m_{ij}$. 
For brevity, we will denote these two half-edges by $h_{ij}$ and $h_{ji}$, respectively. 
Clearly, each half-edge has length at most the radius of the disk it lies in,
and so $h_{ij}\subset D_i$ and $h_{ji}\subset D_j$. 
We denote the resulting set of half-edges by~$\E_1$. 

We label each crossing $x\in \X$ with an unordered pair of integers $\{\lambda_1(x),\lambda_2(x)\}$,
defined as follows: if $x$ is the crossing between the half-edges $h_{ij},h_{k\ell}$, 
then $\lambda_1(x)$ is the number of crossings contained in
$\pi_{ij}[x,m_{ij}]$ and $\lambda_2(x)$ is the number of crossings contained in $\pi_{k\ell}[x,m_{k\ell}]$; see Fig.~\ref{fig:labeling}(i).
This labeling is useful to obtain a rough bound on the total ply of all crossings, 
because of the following observation, which immediately follows from the triangle inequality.
%---------------------------------------------------------------------------------------
\obsinproof{1}{Consider a crossing $x= h_{ij} \cap h_{k\ell}$. If
$d(x,m_{k\ell})\leq  d(x,m_{ij})$ then
all crossings $y\in \pi_{k\ell}[x,m_{k\ell}]$ are contained in $D_i$, 
and otherwise all crossings $y\in \pi_{ij}[x,m_{ij}]$ are contained in $D_k$.}
%---------------------------------------------------------------------------------------
Let $K \mydef \sum_{x\in \X} \ply(x)$ denote the total ply of all crossings. 
The following claim bounds $K$ in terms of the labels $\{\lambda_1(x),\lambda_2(x)\}$.
%---------------------------------------------------------------------------------------
\claiminproof{1}{$K > \frac{1}{2|\D_1|} \sum_{x\in \X} \min\{\lambda_1(x),\lambda_2(x)\}$.}
%---------------------------------------------------------------------------------------
\begin{proofinproof}
Define $K(D_i) \mydef \left|\strut\{ x\in\X : x\in D_i \}\right|$ to be the contribution of $D_i$ to the total
ply~$K$, and note that 
\[
K = \sum_{x\in\X} \ply(x) = \sum_{x\in\X} \left|\strut\{ D_i \in \D_1 : x\in D_i \}\right| 
   = \sum_{D_i \in \D_1} \left|\strut\{ x\in \X : x\in D_i \}\right| = \sum_{D_i\in\D_1} K(D_i) .
\]
Now consider a crossing $x=h_{ij}\cap h_{k\ell}$. If $d(x,m_{k\ell})\leq  d(x,m_{ij})$ 
then we assign $x$ to~$D_i$, and otherwise we assign $x$ to $D_k$. Let $\X(D_i)$
be the set of crossings assigned to~$D_i$. By Observation~1 and the definition of the label $\{\lambda_1(x),\lambda_2(x)\}$,
the disk~$D_i$ contains at least $\min\{\lambda_1(x),\lambda_2(x)\}$ crossings
$y\in \pi_{k\ell}[x,m_{k\ell}]$, for any $x=h_{ij}\cap h_{k\ell}\in \X(D_i)$.
Thus, summing over all crossings 
$x\in\X(D_i) \cap h_{ij}$ and all half-edges $h_{ij}$ incident to~$D_i$, we find that
\[
K(D_i) 
\geq \frac{1}{2\deg(D_i)} \sum_{h_{ij}} \sum_{x\in \X(D_i)\cap h_{ij}} \min\{\lambda_1(x),\lambda_2(x)\}
> \frac{1}{2|\D_1|} \sum_{x\in \X(D_i)} \min\{\lambda_1(x),\lambda_2(x)\},
\]
where $\deg(D_i)$ denotes the degree of $D_i$ in $\ig(\D_1)$.
The factor $\frac{1}{2\deg(D_i)}$ arises because a crossing $y\in h_{k\ell}$ 
can be counted up to $2\deg(D_i)$ times in the expression $\sum_{x\in \X(D_i)} \min\{\lambda_1(x),\lambda_2(x)\}$, namely at most twice for every half-edge 
incident to $D_i$; see Fig.~\ref{fig:labeling}(ii). (Twice, because a pair of
paths in a proper path set may cross twice.)
Since each crossing is assigned to exactly one set $\X(D_i)$, we obtain
\begin{align*}
K &= \sum_{D_i\in\D_1} K(D_i) 
  > \sum_{D_i\in \D} \left( \frac{1}{2|\D_1|} \sum_{x\in \X(D_i)} \min \{\lambda_1(x),\lambda_2(x)\} \right)\\
  &=  \frac{1}{2|\D_1|} \sum_{x\in \X} \min\{\lambda_1(x),\lambda_2(x)\} . 
\end{align*}
\end{proofinproof}
From the Crossing Lemma and our initial assumption that 
$|E_1|>\sqrt{\frac{4}{c}}\cdot n^{\frac{3+\alpha}{2}}$, we have that
\begin{equation}
|\X| \geq c \frac{|E_1|^3}{n^2} >4|E_1| n^{1+\alpha} = 2|\E_1| n^{1+\alpha}. \label{eq:X}
\end{equation}
In order to get a rough bound for $\sum_{x\in\X} \min\{\lambda_1(x),\lambda_2(x)\}$, 
we will ignore crossings with small labels, while ensuring that we don't ignore too 
many crossings in total. More precisely, for every half-edge $h_{ij}$, we disregard
its first  
$n^{1+\alpha}$ crossings, starting from the one closest to~$m_{ij}$. 
We let $\X^*$ denote the set of remaining crossings. 
Note that $|\X^*| \geq |\X|-|\E_1|n^{1+\alpha}$, and $\min\{\lambda_1(x),\lambda_2(x)\} \geq n^{1+\alpha}$ for every $x\in \X^*$. Therefore 
\[
K  > \frac{1}{2|\D_1|} \sum_{x\in \X^* } \min\{\lambda_1(x),\lambda_2(x)\} 
    = \frac{1}{2|\D_1|} \cdot |\X^*| \cdot n^{1+\alpha}
    \geq  \frac{\left(|\X|-|\E_1|n^{1+\alpha}\right)n^{1+\alpha}}{2|\D_1|}
    \geq  \frac{|\X|}{4}n^{\alpha}.
\]
This means that there exists a crossing $x\in \X$ of ply at least $\frac{1}{4}n^{\alpha}$, which contradicts the condition of the lemma and thus
finishes the proof. 
\end{proof}
%---------------------------------------------------------------------------------------

%---------------------------------------------------------------------------------------
\mypara{Step~3: Applying the separator theorem for string graphs.}
%---------------------------------------------------------------------------------------
Lee's separator theorem for string graphs \cite{Lee-string-sep} states that any string
graph on $m$ edges admits a balanced separator of size $O(\sqrt{m})$. 
It is well known (and easy to show) that any intersection graph of
path-connected sets in the plane is a string graph. Hence, we can apply
Lee's result to $\ig(\D_1)$ which, as we have just shown, has $O(n^{\frac{3+\alpha}{2}})$
edges.
Thus, $\ig(\D_1)$ has a separator of size $O(n^{\frac{3+\alpha}{4}})$. If we add the vertices
of this separator as singletons to our clique-based separator $\calS$ then, together with
the cliques added in Step~1, we obtain a separator consisting of $O(n^{\frac{3+\alpha}{4}} + n^{1-\alpha})$ cliques. Picking $\alpha=1/5$, and anticipating the
extension to the case where $\Pi(\D_1)$ is not a proper path set,
we obtain the following result.
%---------------------------------------------------------------------------------------
\begin{proposition}
Let $d$ be a well-behaved shortest-path metric on a closed and path-connected subset $F\subset\Reals^2$ and let 
$\D$ be a set of $n$ geodesic disks with respect to the metric~$d$. 
Then $\ig(\D)$ has a balanced clique-based separator consisting of $O(n^{4/5})$ cliques.
\end{proposition}
%---------------------------------------------------------------------------------------

%---------------------------------------------------------------------------------------
\subsection{Bootstrapping} 
\label{sec:improving}
%---------------------------------------------------------------------------------------
We now describe a bootstrapping mechanism to improve the size of our clique-based separator. 
We first explain how to apply the mechanism once, to obtain a clique-based separator 
consisting of $O(n^{10/13})$ cliques. Then we apply the method repeatedly to obtain a separator
consisting of $O(n^{3/4+\varepsilon})$ cliques.

%---------------------------------------------------------------------------------------
\subparagraph{The basic bootstrapping mechanism.}
%---------------------------------------------------------------------------------------
The idea of our bootstrapping mechanism is as follows. 
After reducing the ply in Step~1, we argued that the number of remaining edges is subquadratic. 
We can actually reduce the number of edges and crossings even further, by reducing the maximum degree 
of the graph $\ig(\D_1)$. To this end, after reducing the ply, we remove from $\ig(\D_1)$ all 
vertices~$D_i$ such that $\deg(D_i)\geq n^\beta$, for some constant $0<\beta<1$ whose value will 
be determined later. All these high degree vertices are placed in our clique-based separator as singletons. 
If $h$ is the number of these vertices, then $h\cdot n^\beta \leq 2 |E_1|$ and so $h= O( n^{\frac{3+\alpha}{2}-\beta})$. 
We denote the resulting graph by $\ig(\D_2) = (\D_2,E_2)$ and let $\E_2$ be the half-edges corresponding to the edges in $E_2$.

To minimize the number of cliques in our separator, the constants $\alpha$ and $\beta$ should be chosen to satisfy the equation 
\begin{equation}
\frac{3+\alpha}{2}-\beta=1-\alpha.  \label{eq;alpha-vs-beta} 
\end{equation}
%---------------------------------------------------------------------------------------
Next, we show that after reducing the maximum degree, the number of edges has also decreased. 
%---------------------------------------------------------------------------------------
\begin{lemma}\label{subquadratic-better}
    For the graph $\ig(\D_2)$ defined as above, it holds that $|E_2| \leq \sqrt{\frac{4}{c}} n^{1+\frac{\alpha+\beta}{2}}$, where $c$ is the constant appearing in the Crossing Lemma.
\end{lemma}

\begin{proof}
    The proof is essentially the same as the proof of Lemma \ref{subquadratic}, so we only mention the key differences. 
    We assume for a contradiction that $|E_2| > \sqrt{\frac{4}{c}} n^{1+\frac{\alpha+\beta}{2}}$. Then, we lower bound the number of crossings (denoted now by $\X_2$) as 
    \[
|\X_2| \geq c \frac{|E_2|^3}{n^2} >2|\E_2|n^{\alpha+\beta}.
\]

We then proceed with the same labeling procedure as in Step 2 in order to lower bound the total ply $K_2$ of all crossings. A key difference is that now we only need to divide by $2n^\beta$, since $n^\beta$ is the maximum degree of $\ig(\D_2)$: 
\[
K_2 \geq \frac{1}{2n^\beta} \sum_{x\in \X_2} \min\{\lambda_1(x),\lambda_2(x)\}.
\]

To get an estimate for $K_2$, we remove from every half-edge, its first $n^{\alpha+\beta}$ crossings. In this way there are at least $|\X_2|-|\E_2| n^{\alpha+\beta}$ crossings remaining, each having a minimum label larger than $n^{\alpha+\beta}$. Therefore 
\[
K_2  >  \frac{\left(|\X_2|-|\E_2|n^{\alpha+\beta}\right)n^{\alpha+\beta}}{2n^\beta} 
    >  \frac{|\X_2|}{4}n^{\alpha}
\]
so we again get a contradiction, since there has to exist a crossing with ply at least~$\frac{1}{4}n^{\alpha}$.
\end{proof}
%---------------------------------------------------------------------------------------
Step 3 now gives us that $\ig(\D_2)$ has a separator of size $O(n^{\frac{2+\alpha+\beta}{4}})$, 
whose vertices we add to our clique-based separator as singletons. In order to minimize the number of cliques in our separator,
we now have the equation
\begin{equation}
\frac{2+\alpha+\beta}{4}=1-\alpha. \label{eq2}
\end{equation}
Solving the system of Equations~(\ref{eq;alpha-vs-beta}) and~(\ref{eq2}) gives $\alpha=3/13$ and $\beta=11/13$. 
This solution corresponds to a clique-based separator of size $O(n^{10/13})$. 

%---------------------------------------------------------------------------------------
\subparagraph{Repeated bootstrapping.}
%---------------------------------------------------------------------------------------
The method we just described can be applied repeatedly: After reducing the 
maximum degree to $n^{\beta}$, we showed that $|E_2|=O(n^{1+\frac{\alpha+\beta}{2}})$. 
This allows us to reduce the maximum degree even more, by placing 
vertices of degree at least $n^\gamma$ in the clique-based separator,  for some~$\gamma<\beta$. 
If we repeat this process $O(\log (1/\eps))$ times, we can obtain a clique-based separator of size $O(n^{3/4+\eps})$,
as is shown in the following theorem. 
%---------------------------------------------------------------------------------------
\begin{theorem} \label{thm:separator}
     Let $d$ be a well-behaved shortest-path metric on a closed and path-connected subset $F\subset\Reals^2$, 
     and let $\D$ be a set of $n$ geodesic disks with respect to the metric~$d$. 
    Then $\ig(\D)$ has a balanced clique-based separator consisting of $O(n^{\frac{3}{4}+\varepsilon})$ cliques.
\end{theorem}
%---------------------------------------------------------------------------------------
\begin{proof}
    Let $\sep$ denote our clique-based separator. As before, we start by reducing the ply of the subdivision, 
    by finding points of ply at least $\frac{n^{\alpha_1}}{4}$, for some constant $0<\alpha_1<1$,
    and placing the corresponding cliques into~$\sep$. Our separator now has a size of $O(n^{1-\alpha_1})$. Let $\D_1$ denote the remaining set of disks and let $\ig(\D_1)=(\D_1,E_1)$ denote the resulting subgraph after this step. Next, we repeatedly apply the following procedure for $i=2,\dotsc,k$: we reduce the maximum degree of $\ig(\D_{i-1})$ to $n^{\alpha_i}$ (for some positive constant $\alpha_i<
    \alpha_{i-1}$) by removing from $\ig(\D_{i-1})$ vertices of degree at least $n^{\alpha_i}$ and we place them in $\sep$ as singletons. We denote by $\D_i$ the remaining set of disks and let $\ig(\D_{i})=(\D_i,E_i)$ be the resulting subgraph. 
    
    Lemma~\ref{subquadratic-better} gives us that $|E_2| \leq \sqrt{\frac{4}{c}} n^{1+\frac{\alpha_{1}+\alpha_2}{2}}$, for $i=2,\dotsc,k$. 
    In exactly the same way we can show that $|E_{i}| \leq \sqrt{\frac{4}{c}} n^{1+\frac{\alpha_1+\alpha_{i}}{2}}$, for $i=3,...,n$.
    
    During the above procedure, we have introduced $k$ unknown constants $\alpha_1,\alpha_2,\dotsc,\alpha_k$ and so we need $k$ equations in order to calculate their values. To this end, let $h_{i}$ denote the number of vertices of $G_i$ with degree at least $n^{\alpha_{i+1}}$, for $i=1,\dotsc,{k-1}$. Then, due to Lemma~\ref{subquadratic}, we have 
    $
     h_{1} \cdot n^{\alpha_{2}}\leq 2|E_1|, 
    $
    and so 
    \[
    h_1 = O(n^{\frac{3+\alpha_1}{2}-\alpha_2}).
    \]
    For $i=2,...,k-1$ we have 
    $
    h_{i} \cdot n^{\alpha_{i+1}}\leq 2|E_i |,
    $
    and so 
    \[
    h_i = O(n^{1+\frac{\alpha_1+\alpha_i}{2}-\alpha_{i+1}}).
    \]
    In the last step, we are left with the graph $\ig(\D_k)$, for which we have that 
    $|E_k| \leq \sqrt{\frac{4}{c}} n^{1+\frac{\alpha_{1}+\alpha_k}{2}}$. 
    We now apply the separator result on string graphs, which states that any string graph
    with $m$ edges has a separator of size $O(\sqrt{m})$. Thus $\ig(\D_k)$ has a separator of size 
    $O(n^{\frac{2+\alpha_{1}+\alpha_k}{4}})$ whose vertices we also place in $\sep$. 
    As a result, the size of our clique-based separator is given by 
    \[
    |\sep| = O \left( n^{1-\alpha_1}+ \sum_{i=1}^{k-1} h_k + n^{\frac{2+\alpha_{1}+\alpha_k}{4}} \right) 
    \]
    which gives the following system of $k$ equations:
    \begin{subequations}
     \label{eq:bunch}
     \renewcommand{\theequation}{\theparentequation.\arabic{equation}}
    \begin{align}
      1-\alpha_1   & =  \frac{3+\alpha_1}{2}-\alpha_2  \\
       1-\alpha_1   & =  1+\frac{\alpha_1+\alpha_i}{2}-\alpha_{{i+1}} \quad \text{for} \quad i=2,3,...,k-1 \tag{\ref{eq:bunch}.$i$}\\
      1-\alpha_1   & =  \frac{2+\alpha_{1}+\alpha_k}{4}
      \tag{\ref{eq:bunch}.$k$}
      \label{eq:k}
    \end{align}
    \end{subequations}
    This system can be rewritten as 
    \begin{subequations}
     \renewcommand{\theequation}{\theparentequation.\arabic{equation}}
    \begin{align}
       3\alpha_1  & = 2\alpha_2 -1  \tag{\ref{eq:bunch}.1} \\
      3\alpha_1  & = 2\alpha_{i+1}-\alpha_i \quad \text{for} \quad i=2,3,...,k-1 \tag{\ref{eq:bunch}.$i$} \\
      5\alpha_1   & = 2-\alpha_k 
      \tag{\ref{eq:bunch}.$k$}
    \end{align}
    \end{subequations}

    By multiplying the $i$-th equation above with $2^{i-1}$ and summing up, we can see that the terms $\alpha_2,...,\alpha_k$ 
    cancel and we are left with the equation 
    \[
    \alpha_1\left(3\sum_{i=0}^{k-2} 2^i+5\cdot2^{k-1}\right)= 2^k-1,
    \]
    which solves to 
    \[
    \alpha_1 = \frac{2^k-1}{2^{k+2}-3}.
    \]
    Thus $1-\alpha_1 = \frac{3\cdot 2^k -2}{4\cdot 2^k -3}$ (which converges to $\frac{3}{4}$ as $k\rightarrow \infty$).
    By choosing $k= \log_2(1/\varepsilon)$, we have 
    \[
    1-\alpha_1 = \frac{3-2\varepsilon}{4-3\varepsilon},
    \]
    which is smaller than $\frac{3}{4}+\varepsilon$ for $\varepsilon<1$.
    (We could pick $k$ even larger, say $k=\log n$, but the resulting bounds will be only marginally better and look more ugly.)
    Therefore, for these choices of $\alpha_1$ and $k$, our clique-based separator has a size of $O(n^{\frac{3}{4}+\varepsilon})$ as claimed.
\end{proof}
%---------------------------------------------------------------------------------------
De~Berg~\etal~\cite{bkmt-cbsgis-23} showed that if a class of graphs admits a clique-based
separator consisting of $S(n)$ cliques, and the separator can be constructed in polynomial time, 
then one can solve $q$-\col in $2^{O(S(n))}$ time. Note that if we can compute the boundaries
$\bd D_i$ in polynomial time,\footnote{Typically the time to do this would not only depend
on $n$, but also on the complexity of $F$ and the distance function~$d$. For simplicity we state
our results in terms of $n$ only. (This, of course, puts restrictions on the complexity of $F$ and~$d$.)} 
then we can also compute our separator in polynomial time.
Indeed, we can then compute $\ig(\D)$ and the arrangement of the disk boundaries in polynomial time,
which allows us to do Step~1 (reducing the ply) in polynomial time. Since a separator
for string graphs can be computed in polynomial time~\cite{Lee-string-sep}, this is easily seen to imply
that the separator construction runs in polynomial time.
Thus we obtain the following result.
%---------------------------------------------------------------------------------------
\begin{corollary}
Let $d$ be a shortest-path metric on a connected subset $F\subset\Reals^2$ and let 
$\D$ be a set of $n$ geodesic disks with respect to the metric~$d$, where $d$ is such that
the boundaries $\bd D_i$ of the disks in $\D$ can be computed in polynomial time.
Let $q\geq 1$ and $\eps>0$ be fixed constants. Then $q$-\col
can be solved in  $2^{O(n^{\frac{3}{4}+\varepsilon})}$ time on $\ig(\D)$.   
\end{corollary}
%---------------------------------------------------------------------------------------

%---------------------------------------------------------------------------------------
\subsection{Obtaining a proper path system}\label{se:proper-path-system}
%---------------------------------------------------------------------------------------
So far we assumed that $\Pi(\D)$ is a proper path set. 
In general this may not be the case, since the paths in $\Pi(\D)$ may overlap along 1-dimensional 
subpaths and a pair of paths can meet multiple times. (The latter can happen since shortest paths 
need not be unique). Lemma~\ref{lem:perturb} will allow us to still work with a proper path set
in our proof. The proof of this lemma requires two technical assumptions on our shortest-path metric $d$. We thus now introduce the concept of a \emph{well-behaved shortest-path metric}.

%---------------------------------------------------------------------------------------
\begin{subparagraph}{Well-behaved shortest-path metrics.} 
%---------------------------------------------------------------------------------------
We assume that for any finite set $P=\{ p_1,\ldots,p_n\}$ of points in $F$, there exists 
a set $\Pi(P) = \{ \pi_{ij} : 1 \leq i<j \leq n \}$ of shortest paths between the points 
in $P$ with the following properties:
\begin{enumerate}
\item For any two paths $\pi_{ij},\pi_{k\ell}\in \Pi(P)$, the intersection 
      $\pi_{ij}\cap \pi_{k\ell}$ consists of finitely many connected components.
\item Let $P^+\supset P$ be the set of endpoints of these components, and let 
      $\Gamma_{ij}$ be the set of subpaths into which $\pi_{ij}$ is partitioned by 
      the points in $P^+$. Note that any two subpaths $\gamma\in\Gamma_{ij}$ and 
      $\gamma'\in\Gamma_{k\ell}$ either have disjoint interiors or they are identical. 
      Let $\Gamma = \bigcup_{i,j} \Gamma_{ij}$ be the set of all distinct subpaths. 
      Then there are values $\delta_1,\delta_2>0$ such that the following holds. 
\begin{itemize}
\item For any $v\in P^+$, the Euclidean ball $B(v,\delta_1)$ of radius $\delta_1$ 
      centered at $v$ only intersects paths from $\Pi(P)$ that contain~$v$, and 
      inside $B(v,\delta_1)$ any two paths are either disjoint (except at $v$) or they coincide.
      Moreover, the balls $B(v,\delta_1)$ for $v\in P^+$ are pairwise disjoint.
\item Outside the balls $B(v,\delta_1)$, the minimum distance between any two subpaths $\gamma,\gamma'\in \Gamma$ is at least $\delta_2$.
\end{itemize}
\end{enumerate}
We call a metric satisfying these conditions \emph{well-behaved}. The conditions are 
easily seen to be satisfied when $F$ is a closed and connected polygonal region and 
the paths $\pi_{ij}$ are shortest (in the Euclidean sense) paths inside~$F$. 
In fact, $F$ does not have to be polygonal; its boundary may consist of finitely 
many constant-degree algebraic curves. Another example is when the paths in $\pi_{ij}$ 
are projections onto $\Reals^2$ of shortest paths on a polyhedral terrain, or when 
the paths are shortest paths in a weighted polygonal subdivision.  (In the latter case, 
the length of a section of the path inside a region is multiplied by the weight of the 
region~\cite{weighted}. 
In this setting, shortest paths are piece-wise linear and so they satisfy the conditions.)
\end{subparagraph}
%---------------------------------------------------------------------------------------

We can now proceed with the proof of the lemma. It might be known,
but we have not been able to find a reference. Recall that $E$ is the edge set of $\ig(\D)$.
%---------------------------------------------------------------------------------------
\begin{lemma} \label{lem:perturb}
Let $d$ be a well-behaved shortest-path metric on a closed and path-connected subset $F\subset \Reals^2$,
and let $\Pi(\D) = \{ \pi_{ij} : (D_i,D_j)\in E\}$ be a set of shortest paths 
with the properties stated above. Then there is a
proper path set $\Pi'(\D) = \{\pi'_{ij}: (D_i,D_j)\in E\}$ in $\Reals^2$ with the following
property: for every pair $\pi'_{ij},\pi'_{k\ell}$ of crossing paths in $\Pi'(\D)$, 
the corresponding paths $\pi_{ij},\pi_{k\ell}$ in $\Pi(\D)$ intersect.
\end{lemma}
%---------------------------------------------------------------------------------------
\begin{proof}
We will prove the lemma in two steps. First, we modify the paths
such that any two paths $\pi_{ij}$ and $\pi_{k\ell}$ meet in a single connected component, and
then we perturb these paths slightly to obtain a set $\Pi'(\D)$ where any two paths cross in at most two points.

The first modification is done as follows. 
For each pair $\pi_{ij},\pi_{k\ell}\in \Pi(\D)$, and for each connected 
component~$I\subset\pi_{ij}\cap \pi_{k\ell}$, we add both endpoints of~$I$ to the set~$P$,
thus (by the first property of  well-behaved shortest-path metrics) obtaining a finite set~$P^+$.
The set~$P^+$, together with the pieces of the paths~$\pi_{ij}$ connecting them,
forms a plane graph~$\G^+=(P^+,E^+)$. We now slightly modify the edge lengths
in~$\G^+$, to ensure that there is a unique shortest path between any two points in~$P^+$,
and we replace each original path $\pi_{ij}$ by the shortest path between $p_i$ and $p_j$ in $\G^+$;
see Fig.~\ref{fig:roundabout}(i) for an example.
\begin{figure}[h]
\begin{center}
\includegraphics{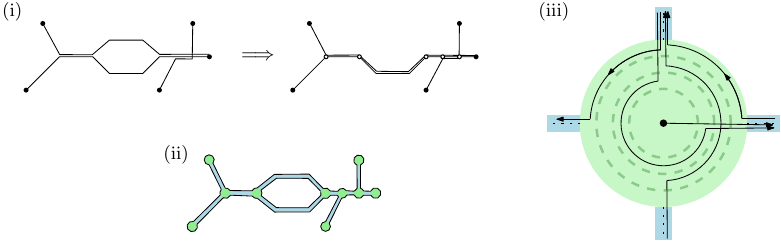}
\end{center}
\caption{(i)~The input paths (left) and the graph $\G^+$ after re-routing the paths (right). 
         (ii)~The thickened graph. 
         (iii)~Drawing the paths inside a thick vertex.} 
\label{fig:roundabout}
\end{figure}
Since the shortest paths between any two vertices in~$\G^+$ are unique,
two paths $\pi_{ij},\pi_{k\ell}$ do not intersect in more than one connected component.
Moreover, the new paths are still shortest paths under the original metric.
Formally, the perturbation of the edges lengths is done as follows.
Let $e_1,e_2,...,e_k$ be any ordering on the edges of~$\G^+$. 
Then, perturb the edge lengths such that the new lengths are given by $|e_i'|=|e_i|+2^i\varepsilon$, 
for an arbitrarily small~$\varepsilon>0$. Any path that was shortest initially, 
will remain shortest and no two paths have the same length. 

Next we show how to slightly perturb the paths such that all crossings are proper. 
For that, we will make use of the second property of well-behaved shortest-path metrics to 
thicken the graph $\G^+$. Intuitively, we replace the edges of $\G^+$ by thick curves 
and its vertices by small balls. Within each thick curve, we can slightly separate all paths 
that use the corresponding edge. Curves going to the same vertex meet at the disk around it. 
We then show that it's possible to route paths within these balls, 
such that paths only possibly cross within the balls where they meet for the first time or 
for the last time. Formally this is done as follows.

Let $\delta_1,\delta_2>0$ be the values specified in the second property of well-behaved metrics. 
We thicken $\G^+$ by replacing each vertex $v\in P^+$ by the ball $B(v,\delta_1)$---we
call this ball a \emph{thick vertex}---and
by replacing each edge in $\G^+$ by a \emph{thick edge} of width~$2\delta_2/3$.
(Formally, we take the Minkowski sum of the original edge and a ball of radius~$\delta/3$,
and remove the parts of the resulting thick edge inside the balls $B(v,\delta_1)$.)
See Fig.~\ref{fig:roundabout}(ii) for an example.
Note that the thick edges are pairwise disjoint.

Now we draw the paths $\pi_{ij}\in \Pi'(\D)$ one by one, as follows. We start
at the point $p_i$ and follow the thick edges of $\pi_{ij}$ until we reach $p_{ij}$.
An invariant of our process is that the paths may only cross inside thick vertices, 
and not inside thick edges. It is easy to see that this invariant can be maintained:
when $\pi_{ij}$ enters a thick edge in between two previously drawn
paths, then we can stay in between these two paths until we reach the next                          
thick vertex. It will be convenient to view the edges as roads
with two lanes, where we will always stay in the right lane as we draw the paths.
(Thus the lanes that we use depend on whether we draw $\pi_{ij}$ from $p_i$
to $p_j$ or in the other direction; this choice can be made arbitrarily.)

Routing inside a thick vertex $B(v,\delta_1)$ is done as follows. Let $\deg(v)$ denote the degree of vertex $v$
in~$\G^+$. We view $B(v,\delta_1)$ as a roundabout with $\deg(v)-1$ lanes $L_1,\ldots,L_{\deg(v)-1}$,
where $L_1$ is the outermost lane. Now suppose we draw a path $\pi_{ij}$ and 
we enter the roundabout through some thick edge. Then
draw $\pi_{ij}$ as follows: if $\pi_{ij}$ leaves the roundabout at the $j$-th exit in counterclockwise
order (counting from the road on which we enter) then we will use lane~$L_j$ on the roundabout
to draw $\pi_{ij}$; see Figure~\ref{fig:roundabout}(iii) for some examples.
We can easily do this in such a way that we maintain the following invariant: 
paths entering and exiting the roundabout on the same roads do not cross.
Note that paths that use the same two edges incident to $v$ but that
were drawn in the opposite direction, do not cross either. Thus two paths 
$\pi_{ij},\pi_{k\ell}$ may only cross when $v$ is either the first or the last vertex
of $\pi_{ij}\cap \pi_{k\ell}$, and when this happens they have
a single crossing inside $B(v,\delta_1)$. Hence, two paths have at most two crossings
and these crossings are proper crossing. 

There is one thing that we have swept under the rug so far, namely what
happens at the start vertex (or end vertex) of the path~$\pi_{ij}$.
In this case we simply draw $\pi_{ij}$ from $p_i$ to the correct exit road 
(or we draw $\pi_{ij}$ from the entry road to $p_j$).
This way we may have to create intersection with paths $\pi_{kl}$ traversing the
roundabout, but this is okay since this is then the first (or last) vertex
where $\pi_{ij}$ meets~$\pi_{kl}$.

\end{proof}
%---------------------------------------------------------------------------------------
Recall that the assumption that $\Pi(\D)$ is a proper path set was used so that
we could apply the Crossing Lemma. Now, instead of applying the Crossing Lemma to
(a subset of) $\Pi(\D)$ we apply it to (the corresponding subset of) $\Pi'(\D)$. 
Then we obtain a bound on $|\X|$, the number of crossings of the perturbed paths, 
which by Lemma~\ref{lem:perturb} gives us a collection of points on the intersections of the original paths.
These points can be used as the set $\X$ in the proof of Lemma~\ref{subquadratic}.
Note that Observation~1 in that proof still holds. Also note that 
the fact that crossings may now coincide is not a problem:
we just need to re-define $\lambda_1(x)$ (and, similarly, $\lambda_2(x)$)
so that crossings that coincide with $x$ are also counted.
We conclude that the proof also works without the assumption that $\Pi(\D)$ is a proper path set.

%-------------------------------------------------------------------------
\section{Application to distance oracles}
%-------------------------------------------------------------------------
%\ba{do we need to mention when the recursion bottom out and what to do then?}
%\mdb{Don't know, we can just stop when $|\D|=1$, should we bother writing this?}
%\ba{maybe not if it's entirely trivial}
With our clique-based separator at hand, we can apply standard techniques
to obtain an almost-exact hop-distance oracle for $\ig(\D)$. This is
done using a \emph{separator tree}~$T$ as follows.
\begin{itemize}
\item Let $\sep$ be a clique-based separator for~$\ig(\D)$. The root of the
      separator tree~$T$ stores, for each disk~$D_i\in\D$ and each clique
      $C_k\in \sep$, the value $\hdist(D_i,C_k) = \min_{D_j\in C_k} \hdist(D_i,D_j)$,
      where $\hdist(D_i,D_j)$ denotes the hop-distance between $D_i$ and $D_j$ in $\ig(\D)$.
      Thus we need $O(n \cdot n^{3/4+\eps})=O(n^{7/4+\eps})$ storage at the root.
\item Let $\D_A$ and $\D_B$ be the two subsets into which $\sep$ splits $\D\setminus \sep$,
      such that $|\D_A|,|\D_B|\leq 2n/3$ and there are no edges between $\D_A$ and $\D_B$.
      We recursively construct separator trees $T_A$ for $\D_A$ and $T_B$ for $\D_B$,
      which become the two subtrees of the root of~$T$. 
\end{itemize}
The amount of storage of the structure satisfies the recurrence 
$S(n) = O(n^{7/4+\eps}) + S(n_A) + S(n_B)$, where $n_A, n_B \leq 2n/3$  and $n_A+n_B \leq n$.
This solves to  $S(n) = O(n^{7/4+\eps})$.
\medskip

To answer a query for the hop-distance between two disks $D_i,D_j$ we proceed as follows.
First, we determine $Z_1 \mydef \min_{C_k\in\sep} \left( \hdist(D_i,C_k) + \hdist(C_k,D_j) \right)$.
When $D_i$ and $D_j$ do not lie in the same part of the partition, we are done and report $Z_1$.
Otherwise, assuming without loss of generality that $D_i,D_j\in \D_A$, we recursively 
query in~$T_A$, thus obtaining a value $Z_2$, and we report $\min(Z_1,Z_2)$.
Thus the query time satisfies the recurrence $Q(n) = O(n^{3/4+\eps}) + Q(n_A)$,
where $n_A\leq 2n/3$, which solves to $Q(n) = O(n^{3/4+\eps})$. It is easily seen that
our query reports a value $d^*$ such that $d^* \leq \hdist(D_i,D_j) \leq d^* +1$.
\begin{theorem}
Let $d$ be a well-behaved shortest-path metric on a closed and path-connected subset $F\subset\Reals^2$,
and let $\D$ be a set of $n$ geodesic disks with respect to the metric~$d$. Let $\eps>0$
be any fixed constant.
Then there is a distance oracle for $\ig(\D)$ that uses $O(n^{\frac{7}{4}+\varepsilon})$
storage and that can report, for any two query disks $D_i,D_j\in \D$, in $O(n^{3/4+\eps})$ time
a value $d^*$ such that $d^* \leq \hdist(D_i,D_j) \leq d^* +1$.
\end{theorem}

%-------------------------------------------------------------------------
\section{Concluding remarks}
%-------------------------------------------------------------------------
In this paper, we showed that the intersection graph of a set of geodesic disks, in any well-behaved 
shortest-path metric in the plane, admits a clique-based separators of sublinear size, using a method 
that is quite different from previously used approaches.
Separators have been used extensively to obtain efficient graph algorithms, and clique-based separators have already found many
uses as well. We gave two straightforward applications of our new clique-based separators, namely for $q$-\col
and almost exact distance oracles, but we expect our separator to have more applications.
More generally, we hope that our paper inspires more research on intersection graphs of geodesic
disks in the general setting that we studied---or, more modestly, of geodesic disks in polygons with holes
or on terrains. 

An obvious open problem is to improve upon our bounds: do intersection graphs of geodesic disks admit a separator
of size $o(n^{3/4})$, either in general or perhaps in the setting of polygons with holes?
Our distance oracle is the first almost exact distance oracle with sublinear query time and subquadratic 
storage in such a general setting, but (besides our novel separator) it uses only
standard techniques. It would be interesting to see if more advanced
techniques can be applied to get better bounds.

\bibliography{ref}

\begin{thebibliography}{10}

\bibitem{abraham-gavoille}
Ittai Abraham and Cyril Gavoille.
\newblock {On approximate distance labels and routing schemes with affine stretch}.
\newblock In {\em {Proc.~$25^{th}$ International Symposium on Distributed Computing (DISC 2011)}}, volume 6950 of {\em Lecture Notes in Computer Science (ARCoSS)}, pages 404--415, 2011.
\newblock \href {https://doi.org/10.1007/978-3-642-24100-0_39} {\path{doi:10.1007/978-3-642-24100-0_39}}.

\bibitem{AJTAI19829}
M.~Ajtai, V.~Chvátal, M.M. Newborn, and E.~Szemerédi.
\newblock Crossing-free subgraphs.
\newblock In Peter~L. Hammer, Alexander Rosa, Gert Sabidussi, and Jean Turgeon, editors, {\em Theory and Practice of Combinatorics}, volume~60 of {\em North-Holland Mathematics Studies}, pages 9--12. North-Holland, 1982.
\newblock \href {https://doi.org/doi.org/10.1016/S0304-0208(08)73484-4} {\path{doi:doi.org/10.1016/S0304-0208(08)73484-4}}.

\bibitem{DBLP:conf/esa/ArikatiCCDSZ96}
Srinivasa~Rao Arikati, Danny~Z. Chen, L.~Paul Chew, Gautam Das, Michiel H.~M. Smid, and Christos~D. Zaroliagis.
\newblock Planar spanners and approximate shortest path queries among obstacles in the plane.
\newblock In {\em Proc.~4th Annual European Symposium on Algorithms ({ESA} 1996)}, volume 1136 of {\em Lecture Notes in Computer Science}, pages 514--528, 1996.
\newblock \href {https://doi.org/10.1007/3-540-61680-2\_79} {\path{doi:10.1007/3-540-61680-2\_79}}.

\bibitem{socg-version}
Boris Aronov, Mark de~Berg, and Leonidas Theocharous.
\newblock A clique-based separator for intersection graphs of geodesic disks in $\Reals^2$.
\newblock In {\em Proc. 40th International Symposium on Computational Geometry (SoCG 2024)}, LIPIcs. Schloss Dagstuhl - Leibniz-Zentrum f{\"{u}}r Informatik, 2024.

\bibitem{BonnetR19}
{\'{E}}douard Bonnet and Pawel Rzazewski.
\newblock Optimality program in segment and string graphs.
\newblock {\em Algorithmica}, 81(7):3047--3073, 2019.
\newblock \href {https://doi.org/10.1007/s00453-019-00568-7} {\path{doi:10.1007/s00453-019-00568-7}}.

\bibitem{chan-skrepetos}
Timothy~M. Chan and Dimitrios Skrepetos.
\newblock {Approximate Shortest Paths and Distance Oracles in Weighted Unit-Disk Graphs}.
\newblock In {\em Proc.~34th International Symposium on Computational Geometry (SoCG 2018)}, volume~99, pages 24:1--24:13, 2018.
\newblock \href {https://doi.org/10.4230/LIPIcs.SoCG.2018.24} {\path{doi:10.4230/LIPIcs.SoCG.2018.24}}.

\bibitem{DBLP:journals/jacm/CharalampopoulosGLMPWW23}
Panagiotis Charalampopoulos, Pawel Gawrychowski, Yaowei Long, Shay Mozes, Seth Pettie, Oren Weimann, and Christian Wulff{-}Nilsen.
\newblock Almost optimal exact distance oracles for planar graphs.
\newblock {\em J. {ACM}}, 70(2):12:1--12:50, 2023.
\newblock \href {https://doi.org/10.1145/3580474} {\path{doi:10.1145/3580474}}.

\bibitem{DBLP:conf/stoc/Chechik14}
Shiri Chechik.
\newblock Approximate distance oracles with constant query time.
\newblock In {\em Proc.~46th Symposium on Theory of Computing ({STOC} 2014)}, pages 654--663, 2014.
\newblock \href {https://doi.org/10.1145/2591796.2591801} {\path{doi:10.1145/2591796.2591801}}.

\bibitem{b-nrdo-23}
Mark de~Berg.
\newblock A note on reachability and distance oracles for transmission graphs.
\newblock {\em Computing in Geometry and Topology}, 2(1):4:1--4:15, 2023.
\newblock \href {https://doi.org/doi.org/10.57717/cgt.v2i1.25} {\path{doi:doi.org/10.57717/cgt.v2i1.25}}.

\bibitem{bbkmz-ethf-20}
Mark de~Berg, Hans~L. Bodlaender, S{\'{a}}ndor Kisfaludi{-}Bak, D{\'{a}}niel Marx, and Tom~C. van~der Zanden.
\newblock A framework for {Exponential-Time-Hypothesis}-tight algorithms and lower bounds in geometric intersection graphs.
\newblock {\em SIAM J. Comput.}, 49:1291--1331, 2020.
\newblock \href {https://doi.org/10.1137/20M1320870} {\path{doi:10.1137/20M1320870}}.

\bibitem{bkmt-cbsgis-23}
Mark de~Berg, S{\'{a}}ndor Kisfaludi{-}Bak, Morteza Monemizadeh, and Leonidas Theocharous.
\newblock Clique-based separators for geometric intersection graphs.
\newblock {\em Algorithmica}, 85(6):1652--1678, 2023.
\newblock \href {https://doi.org/10.1007/S00453-022-01041-8} {\path{doi:10.1007/S00453-022-01041-8}}.

\bibitem{DV-cycle-separator}
Hristo Djidjev and Shankar~M. Venkatesan.
\newblock Reduced constants for simple cycle graph separation.
\newblock {\em Acta Informatica}, 34(3):231--243, 1997.
\newblock \href {https://doi.org/10.1007/s002360050082} {\path{doi:10.1007/s002360050082}}.

\bibitem{fox2008}
Jacob Fox and János Pach.
\newblock Separator theorems and {T}ur\'{a}n-type results for planar intersection graphs.
\newblock {\em Advances in Mathematics}, 219(3):1070--1080, 2008.
\newblock \href {https://doi.org/doi.org/10.1016/j.aim.2008.06.002} {\path{doi:doi.org/10.1016/j.aim.2008.06.002}}.

\bibitem{FU2011379}
Bin Fu.
\newblock Theory and application of width bounded geometric separators.
\newblock {\em Journal of Computer and System Sciences}, 77(2):379--392, 2011.
\newblock \href {https://doi.org/10.1016/j.jcss.2010.05.003} {\path{doi:10.1016/j.jcss.2010.05.003}}.

\bibitem{gao-zhang}
Jie Gao and Li~Zhang.
\newblock Well-separated pair decomposition for the unit-disk graph metric and its applications.
\newblock {\em SIAM Journal on Computing}, 35(1):151--169, 2005.
\newblock \href {https://doi.org/10.1137/S0097539703436357} {\path{doi:10.1137/S0097539703436357}}.

\bibitem{HQ-low-dens-sep}
Sariel Har{-}Peled and Kent Quanrud.
\newblock Approximation algorithms for polynomial-expansion and low-density graphs.
\newblock {\em {SIAM} J. Comput.}, 46(6):1712--1744, 2017.
\newblock \href {https://doi.org/10.1137/16M1079336} {\path{doi:10.1137/16M1079336}}.

\bibitem{DBLP:conf/focs/LeW21}
Hung Le and Christian Wulff{-}Nilsen.
\newblock Optimal approximate distance oracle for planar graphs.
\newblock In {\em Proc.~62nd {IEEE} Annual Symposium on Foundations of Computer Science ({FOCS} 2021)}, pages 363--374. {IEEE}, 2021.
\newblock \href {https://doi.org/10.1109/FOCS52979.2021.00044} {\path{doi:10.1109/FOCS52979.2021.00044}}.

\bibitem{Lee-string-sep}
James~R. Lee.
\newblock Separators in region intersection graphs.
\newblock In {\em Proc. 8th Innovations in Theoretical Computer Science Conference ({ITCS} 2017)}, volume~67 of {\em LIPIcs}, pages 1:1--1:8, 2017.
\newblock \href {https://doi.org/10.4230/LIPIcs.ITCS.2017.1} {\path{doi:10.4230/LIPIcs.ITCS.2017.1}}.

\bibitem{l-civlsi-03}
Thomas Leighton.
\newblock {\em Complexity Issues in VLSI}.
\newblock Foundations of Computing Series. MIT Press, 2003.

\bibitem{LT-planar-separator-thm}
Richard~J. Lipton and Robert~Endre Tarjan.
\newblock A separator theorem for planar graphs.
\newblock {\em {SIAM} J. Appl. Math}, 36(2):177--189, 1977.
\newblock \href {https://doi.org/doi/10.1137/0136016} {\path{doi:doi/10.1137/0136016}}.

\bibitem{weighted}
Joseph S.~B. Mitchell and Christos~H. Papadimitriou.
\newblock The weighted region problem: Finding shortest paths through a weighted planar subdivision.
\newblock {\em J. ACM}, 38(1):18–73, jan 1991.
\newblock \href {https://doi.org/10.1145/102782.102784} {\path{doi:10.1145/102782.102784}}.

\bibitem{patrascu}
Mihai Patrascu and Liam Roditty.
\newblock Distance oracles beyond the {Thorup-Zwick} bound.
\newblock In {\em Proc. 51st Annual Symposium on Foundations of Computer Science (FOCS 2010)}, pages 815--823, 2010.
\newblock \href {https://doi.org/10.1109/FOCS.2010.83} {\path{doi:10.1109/FOCS.2010.83}}.

\bibitem{DBLP:journals/csur/Sommer14}
Christian Sommer.
\newblock Shortest-path queries in static networks.
\newblock {\em {ACM} Comput. Surv.}, 46(4):45:1--45:31, 2014.
\newblock \href {https://doi.org/10.1145/2530531} {\path{doi:10.1145/2530531}}.

\bibitem{DBLP:journals/jacm/ThorupZ05}
Mikkel Thorup and Uri Zwick.
\newblock Approximate distance oracles.
\newblock {\em J. {ACM}}, 52(1):1--24, 2005.
\newblock \href {https://doi.org/10.1145/1044731.1044732} {\path{doi:10.1145/1044731.1044732}}.

\bibitem{thurston-78}
William Thurston.
\newblock {\em The Geometry and Topology of 3-Manifolds}.
\newblock Princeton Lecture Notes, 1978--1981.

\end{thebibliography}

\newpage

%% For arXiv
% \appendix
% \input{appendix}
\end{document}